\documentclass[
  aps,
  prl,
  reprint,
  superscriptaddress,
  nofootinbib
]{revtex4-2}

\usepackage{amsmath}
\usepackage{amssymb}
\usepackage{amsthm}
\usepackage{graphicx}
\usepackage{booktabs}
\usepackage{bm}
\usepackage{mathtools}
\usepackage{hyperref}

\usepackage[qm]{qcircuit}

\usepackage{tikz}

\usepackage{xcolor}

\newcommand*\qcontrolcolor[2]{%
  \push{%
    \tikz[baseline=(char.base)]{%
      \node[
        shape=circle,
        draw,
        inner sep=1.5pt,
        color=#1
      ] (char) {\scriptsize #2};%
    }%
  }%
  \qw
}

\newcommand*\qcontrolcolour[2]{%
  \push{%
    \tikz[baseline=(char.base)]{%
      \node[
        shape=circle,
        draw,
        inner sep=.6pt,
        color=#1
      ] (char) {\scriptsize #2};%
    }%
  }%
  \qw
}

\newcommand*\onecontrol{%
  \qcontrolcolour{red!80!black}{1}%
}

\newcommand*\twocontrol{%
  \qcontrolcolor{blue!100!white}{2}%
}

\newcommand{\ket}[1]{\lvert #1\rangle}
\newcommand{\bra}[1]{\langle #1\rvert}
\newcommand{\proj}[1]{\lvert #1\rangle\!\langle #1\rvert}

\newcommand{\INC}{\mathrm{INC}}
\newcommand{\SUM}{\mathrm{SUM}}
\newcommand{\MCT}{\mathrm{MCT}}

\newcommand{\Csel}{C_{2}(\INC)}
\newcommand{\Cstrict}{C_{2}(X_{01})}

\newcommand{\optionalfigure}[2]{%
  \IfFileExists{#1}{%
    \includegraphics[width=#2]{#1}%
  }{%
    \fbox{%
      \parbox[c][1.10in][c]{0.90\linewidth}{%
        \centering
        Replace with vector figure\\[0.35em]
        \texttt{\detokenize{#1}}%
      }%
    }%
  }%
}

\newtheorem{theorem}{Theorem}
\newtheorem{proposition}{Proposition}
\newtheorem{lemma}{Lemma}

\begin{document}

\title{
Efficient Synthesis of Multi-Controlled Toffoli Gates with Ternary Clifford$+P_9$ Gates
}

\author{Amit Saha}
\author{Francesco Arzani}
\affiliation{
DI-ENS, École Normale Supérieure, Université PSL, CNRS, INRIA, 45 rue d'Ulm, 75005 Paris, France
}

\date{\today}

\begin{abstract}
Temporary occupation of qutrit levels can reduce the width and depth required for binary circuit logic in quantum computing. We introduce an efficient intermediate-qutrit decomposition of multi-controlled Toffoli gates with binary-subspace inputs and outputs. For balanced control widths $n=2^h-1$, the proposed decomposition is evaluated hierarchically through a tree, allowing independent subtree computations to proceed in parallel and reducing the depth from linear to logarithmic. The balanced low-depth construction requires $6n+3$ logical $P_9$ injections and $\frac{n-3}{4}$ clean ancillary qutrits. It matches the direct $P_9$ count of a recursively extended Clifford+$P_9$ baseline derived from the state-of-the-art work while using asymptotically one-quarter as many clean ancillas and replacing linear depth with logarithmic depth. For arbitrary control widths, a sequential extension preserves the $6n+3$ $P_9$ count but has depth $O(\log m+n-m)$ for a balanced core $m=2^h-1\leq n$, which is linear in the worst case over $n$. By reducing the resource overhead of a widely used reversible primitive, the proposed decomposition provides a practical building block for the design and compilation of more resource-efficient quantum algorithms in the fault-tolerant regime.
\end{abstract}

\maketitle

\section{Introduction}

Multi-controlled Toffoli (MCT) gates are fundamental primitives in
quantum circuit design. They appear in reversible arithmetic \cite{shor},
oracle construction \cite{grover1996fastquantummechanicalalgorithm}, amplitude amplification \cite{Brassard_2002}, quantum search \cite{Magniez_2011}, and
fault-tolerant compilation \cite{FT}. Their decomposition becomes increasingly
expensive as the number of controls grows, since the implementation
must balance circuit depth, non-Clifford resource cost, and ancillary
workspace \cite{Barenco1995}. These tradeoffs are particularly important in
fault-tolerant architectures, where non-Clifford operations generally
require resource-state preparation and injection \cite{gottesman}.

In conventional binary-only fault-tolerant architectures, multi-controlled Toffoli gates are typically decomposed into circuits of smaller Toffoli gates, which are subsequently compiled into the Clifford+$T$ gate set \cite{Selinger2013, Jones2013}. Such decompositions involve a tradeoff among $T$-gate cost, ancillary-qubit requirements, and circuit depth, with the resource overhead generally increasing as the number of controls grows. Recent binary constructions based on conditionally clean ancillae achieve logarithmic-depth MCT synthesis with substantially fewer clean
ancilla qubits~\cite{NieZiSun2024, KhattarGidney2025, Dutta2025}. These dynamic-circuit approaches, however, employ measurement-assisted uncomputation and classical feed-forward. Therefore, the practical overhead can depend strongly on measurement latency, feed-forward cost, coherent-scheduling requirements, and the additional error-correction with mid-circuit measurements. A complementary recent direction shows that, if a small approximation error is allowed, randomized Clifford+$T$ implementations can realize multi-qubit Toffoli gates using only
$O(\log(1/\epsilon))$ $T$ gates~\cite{GossetKothariZhang2025}; this represents a distinct approximate-synthesis regime from the exact construction considered here in this paper. A useful exact construction strategy is to embed binary information into multilevel quantum systems and temporarily access noncomputational levels during the circuit execution \cite{MajumdarHybridFT}. In the qutrit setting, the logical states
remain encoded in the binary subspace
\begin{equation}\nonumber
\mathcal{H}_{\mathrm{bin}}
=
\operatorname{span}
\left\{
\ket{0},
\ket{1}
\right\}
\subset
\mathcal{H}_{3},
\end{equation}
where
\begin{equation}\nonumber
\mathcal{H}_{3}
=
\operatorname{span}
\left\{
\ket{0},
\ket{1},
\ket{2}
\right\}.
\end{equation}
The level
$\ket{2}$
is occupied only transiently and serves as a workspace state that
records whether selected groups of control qutrits are all in the
state $\ket{1}$.
The circuit input and output
states therefore remain confined to the binary subspace, while the
additional qutrit level is used internally to reduce the complexity
of the decomposition.

A complementary fault-tolerant perspective was developed by Bocharov \textit{et al.} within the ternary Clifford+$P_9$ framework~\cite{Campbell2012, Bocharov2017}.
In \cite{Bocharov2017}, they provided an ancilla-free binary-subspace Toffoli emulation requiring
$15P_9$
injections and with a one clean ancilla requiring
$12P_9$
injections.
Their recursive control-extension construction adds one binary
control using one additional clean ancillary qutrit and an incremental
cost of
$6P_9$.
Starting from the ancilla-free Toffoli block, repeated extensions
therefore give an
$n$-control
MCT baseline with cost
$6n+3$
and
$n-2$
clean ancillary qutrits.
Starting from the one-clean-ancilla block gives an alternative
baseline with cost
$6n$
and
$n-1$
clean ancillary qutrits. Because the additional controls are incorporated recursively, the overall circuit depth remains linear in the number of control qubits. 

The objective of the present work is to provide the logarithmic-depth tree structure of intermediate-qutrit decompositions with an exact resource analysis in the ternary Clifford+$P_9$ model.  We also provide an extension to arbitrary control widths and characterize the additional sequential-depth overhead explicitly introduced away from the balanced family. We evaluate the controls through a tree, combining partial results level by level until the circuit determines whether all controls are active, motivated by some earlier works \cite{Gokhale2019, SahaPRA2022}.
At the lowest level of the tree, each leaf block checks a small group of controls and temporarily records whether all of them are active. At each higher level, the circuit combines the results from two smaller groups with one additional control, using a clean ancillary qutrit as temporary workspace. For balanced widths
$n=2^{h}-1$, where $h$ is the height of the tree,
the resulting construction requires
\begin{equation}\nonumber
N_{P_9}^{\mathrm{tree}}(n)
=
6n+3
\end{equation}
logical
$P_9$
injections,
\begin{equation}\nonumber
N_{\mathrm{anc}}^{\mathrm{tree}}(n)
=
\frac{n-3}{4}
\end{equation}
clean ancillary qutrits, and logarithmic depth.
Thus, the construction matches the exact
$P_9$
count of the recursively extended ancilla-free baseline derived from~\cite{Bocharov2017}, while reducing the asymptotic ancillary
requirement by a factor of four and replacing linear depth
with logarithmic depth.
The same tree structure can also use fewer ancillary qutrits by reusing them during the computation, although this requires additional operations.

The remainder of this paper is organized as follows.
Section~II introduces the qutrit primitives and the binary-subspace
computation model.
Section~III presents the intermediate-qutrit MCT decomposition and
proves its correctness.
Section~IV derives the exact Clifford+$P_9$ cost, the logarithmic-depth
schedule, the constructive width--cost frontier, and the extension to
arbitrary control widths.
The same section compares the proposed construction with recursive
ternary baselines derived from
\cite{Bocharov2017}.
The Supplemental Material~\cite{SupplementalMaterial} provides representative circuits, detailed derivations, extensions to arbitrary control widths, and conditional comparisons with ternary and binary-only fault-tolerant implementations.

\section{Background}

The construction uses three elementary operations:
the ternary SUM gate,
a state-selective controlled increment,
and a strict controlled toggle of the binary target.
The first operation is a Clifford gate in the ternary setting.
The latter two operations are non-Clifford and determine the
fault-tolerant resource cost of the decomposition.

\paragraph{Qutrit primitives.}

The ternary increment gate is
\begin{equation}\nonumber
\INC\ket{j}
=
\ket{j+1\!\!\!\pmod 3}.
\label{eq:inc}
\end{equation}

The ternary SUM gate is
\begin{equation}\nonumber
\SUM
=
\Lambda(\INC),
\end{equation}
with action
\begin{equation}\nonumber
\SUM\ket{x,y}
=
\ket{x,x+y\!\!\!\pmod 3}.
\label{eq:sum}
\end{equation}
SUM belongs to the ternary Clifford group \cite{Bocharov2017}.
On binary inputs,
\begin{equation}\nonumber
\SUM\ket{1,1}
=
\ket{1,2}.
\label{eq:sum-marker}
\end{equation}

The state-selective controlled increment is
\begin{equation}\nonumber
C_{\ell}(\INC)\ket{x,y}
=
\ket{
x,
y+\delta_{x,\ell}\!\!\!\pmod 3
},
\qquad
\ell\in\{0,1,2\}.
\label{eq:selective-inc}
\end{equation}
In particular,
\begin{equation}\nonumber
\Csel\ket{2,1}
=
\ket{2,2}.
\label{eq:propagate-marker}
\end{equation}
Unlike SUM,
$\Csel$
is non-Clifford.

Finally, define
\begin{equation}\nonumber
X_{01}
=
\ket{0}\!\bra{1}
+
\ket{1}\!\bra{0}
+
\ket{2}\!\bra{2}.
\label{eq:x01}
\end{equation}
The strict controlled toggle of binary target is
\begin{equation}\nonumber
\Cstrict\ket{x,t}
=
\ket{x}
\otimes
X_{01}^{\delta_{x,2}}
\ket{t}.
\label{eq:strict-toggle}
\end{equation}
The distinction between
$\SUM$,
$\Csel$, and $\Cstrict$
is important for the resource analysis.
The SUM gate is Clifford and therefore contributes no
extra resource
cost in fault-tolerant setting, whereas
$\Csel$ and $\Cstrict$
are non-Clifford. The proposed decomposition is designed to minimize the number of
state-selective controlled increments while allowing independent
branches of the tree to be evaluated in parallel, whereas the number of $\Cstrict$ is always one irrespective of the number of controls of an MCT gate. 

\paragraph{Fault-tolerant Clifford+$P_9$ model.} To implement $\Csel$ and $\Cstrict$ non-Clifford gates in a fault-tolerant setting, we need to use $P_9$ operations, which are the costly resources, whereas ternary Clifford gates
are treated as comparatively inexpensive.
The number of logical
$P_9$
injections therefore provides the principal non-Clifford cost metric
used in this work. The non-Clifford phase gate is
\begin{equation}\nonumber
P_9
=
\omega_9^{-1}\proj{0}
+
\proj{1}
+
\omega_9\proj{2},
\qquad
\omega_9=e^{2\pi i/9}.
\label{eq:p9}
\end{equation}

The required logical costs are
\begin{equation}\nonumber
N_{P_9}[\SUM]
=
N_{P_9}[\SUM^\dagger]
=
0,
\label{eq:sum-cost}
\end{equation}
\begin{equation}\nonumber
N_{P_9}
\left[
\Csel^{\pm1}
\right]
=
3,
\label{eq:selective-cost}
\end{equation}
and
\begin{equation}\nonumber
N_{P_9}
\left[
\Cstrict
\right]
=
15.
\label{eq:strict-cost}
\end{equation}

\section{Proposed Decomposition}

For the ease of understanding, we first describe the construction for the balanced family of control
widths
\begin{equation}\nonumber
n=2^h-1,
\qquad
h\geq2.
\label{eq}
\end{equation}
Here,
$n$
denotes the number of binary-encoded control qutrits.
The main idea is to determine whether all controls are in the state
$\ket{1}$
through a balanced tree of intermediate computations.
Instead of checking the controls sequentially, the circuit first
evaluates several small groups of controls independently.
The results from these groups are then combined level by level.
At the root of the tree, a single qutrit reaches the temporary 
state
$\ket{2}$
if and only if every logical control is active.

The construction uses two types of building blocks.
A leaf block checks three binary controls and stores the result
temporarily in one of the control qutrits.
An internal merge block combines the results from two previously
evaluated subtrees with one additional control.
Each internal merge uses one clean ancillary qutrit initialized in
$\ket{0}$.
After the target operation has been applied, the circuit is reversed
to restore all controls and ancillary qutrits.

A leaf block acts on three binary controls
$u,v,w$:
\begin{equation}\nonumber
L(u,v,w)
=
\Csel_{v,w}
\,
\SUM_{u,v}.
\label{eq:leaf-block}
\end{equation}
Operators are written in right-to-left order.
Therefore, the SUM gate acts first, followed by the
state-selective controlled increment. The purpose of the leaf block is to determine whether all three
controls are active.

\begin{lemma}
For binary inputs,
\begin{equation}
w_{\mathrm{out}}=2
\quad\Longleftrightarrow\quad
u=v=w=1.
\label{eq_leaf}
\end{equation}
\end{lemma}

\begin{proof}
The SUM gate promotes
$v$
to
$\ket{2}$
if and only if
$u=v=1$.
The state-selective gate
$\Csel_{v,w}$
then increments
$w$
if and only if
$v=2$.
Because
$w$
is initially restricted to the binary subspace,
it changes from
$\ket{1}$
to
$\ket{2}$
if and only if all three inputs are equal to one.
Thus,
$w_{\mathrm{out}}=2$
if and only if
$u=v=w=1$.
\end{proof}

The qutrit
$w$
therefore acts as the temporary output marker of the leaf block.
It remains available as an input to the next level of the tree.

\paragraph{Internal merge blocks.}

We next describe how two previously computed subtree markers are
combined with one additional binary control. First, define:
\begin{equation}\nonumber
M_v(x,m,y;a_v)
=
\Csel_{a_v,m}
\,
\Csel_{y,a_v}
\,
\Csel_{x,a_v}.
\label{eq:merge-block}
\end{equation}

Let
$x$
and
$y$
denote the marker qutrits of two child subtrees. $a_v$ is an ancillary qutrit and $m$ denotes the binary middle control.


\begin{lemma}
If
$a_v$
is initialized in
$\ket{0}$,
then
\begin{equation}
m_{\mathrm{out}}=2
\quad\Longleftrightarrow\quad
x=y=2
\ \text{and}\
 m_{\mathrm{in}}=1.
\label{eq_marker}
\end{equation}
\end{lemma}

\begin{proof}
The first two chronological operations increment
$a_v$
once for each active child marker.
Starting from
$a_v=0$,
the ancillary qutrit reaches
$a_v=2$
if and only if
$x=y=2$.
The final state-selective operation increments the binary middle
control if and only if
$a_v=2$.
Because
$m$
is initially binary encoded,
it reaches
$\ket{2}$
if and only if
$x=y=2$
and
$m_{\mathrm{in}}=1$.
\end{proof}

After the merge block has been applied,
$m$
acts as the temporary marker of the larger subtree.
This marker can then be used as an input to the next merge level.
The ancillary qutrit
$a_v$
is not immediately reset.
Instead, its cleanup is deferred until after the target toggle.
The inverse traversal restores
$a_v$
to
$\ket{0}$.
This delayed cleanup is essential for obtaining the low
$P_9$
count of the proposed construction.

Figure~\ref{fig} illustrates the overall
decomposition structure.
The left side of the figure represents the forward tree computation.
Independent leaf blocks can be executed in parallel because they act
on disjoint groups of controls.
At the next levels, independent internal merges can also be executed
in parallel whenever they belong to separate branches.
The strict target toggle is applied only once, after the complete
all controls satisfied condition reaches the root.
The right side of the figure represents the inverse traversal, which
removes all temporary markers and restores every ancillary qutrit and make all the output binary.
Explicit circuit realizations of this construction are
provided in the Supplemental Material~\cite{SupplementalMaterial}
for seven- and fifteen-control MCT gates.
These examples illustrate the parallel evaluation of the leaf blocks,
the propagation of intermediate
$\ket{2}$
markers through the merge levels, the unique strict toggle at the
root, and the inverse traversal that restores all temporary states.

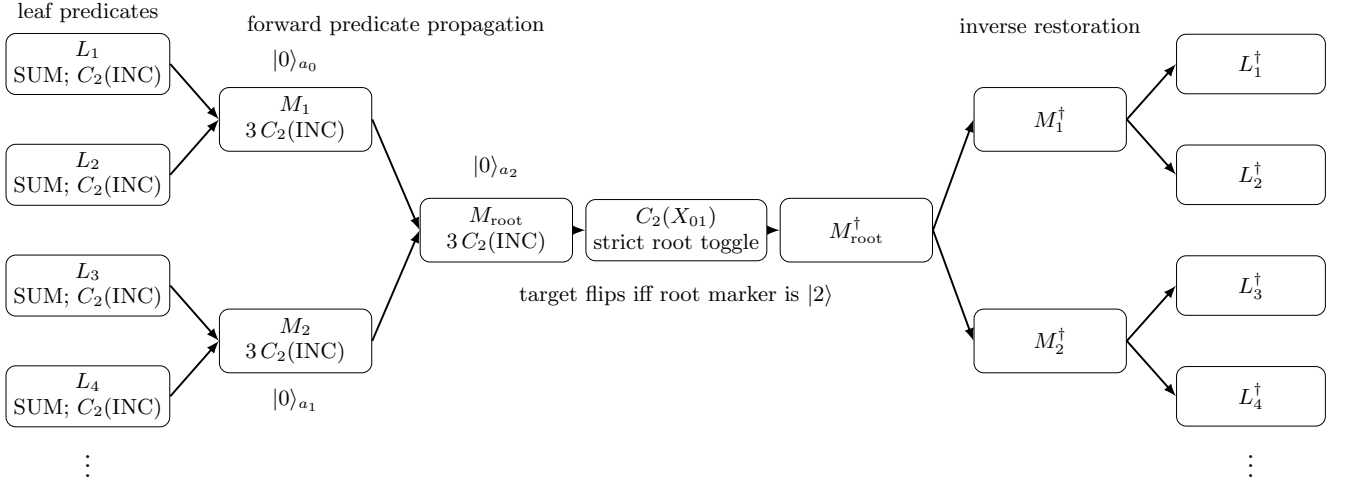
\begin{figure*}[t]
\centering
\resizebox{0.98\textwidth}{!}{%
\begin{tikzpicture}[x=1cm,y=1cm,>=latex]

\tikzset{
  leaf/.style={
    draw,
    rounded corners,
    align=center,
    minimum width=2.15cm,
    minimum height=0.86cm,
    font=\small
  },
  merge/.style={
    draw,
    rounded corners,
    align=center,
    minimum width=2.20cm,
    minimum height=0.92cm,
    font=\small
  },
  roottoggle/.style={
    draw,
    rounded corners,
    align=center,
    minimum width=2.10cm,
    minimum height=0.88cm,
    font=\small
  },
  anclabel/.style={
    font=\small,
    align=center
  },
  flow/.style={
    ->,
    thick
  }
}

\node[leaf] (L1) at (-6.8, 2.4)
{$L_{1}$\\
 $\SUM$; $C_{2}(\INC)$};

\node[leaf] (L2) at (-6.8, 0.8)
{$L_{2}$\\
 $\SUM$; $C_{2}(\INC)$};

\node[leaf] (L3) at (-6.8,-0.8)
{$L_{3}$\\
 $\SUM$; $C_{2}(\INC)$};

\node[leaf] (L4) at (-6.8,-2.4)
{$L_{4}$\\
 $\SUM$; $C_{2}(\INC)$};

\node[merge] (M1) at (-3.8, 1.6)
{$M_{1}$\\
 $3\,C_{2}(\INC)$};

\node[merge] (M2) at (-3.8,-1.6)
{$M_{2}$\\
 $3\,C_{2}(\INC)$};

\node[merge] (MR) at (-0.9,0)
{$M_{\mathrm{root}}$\\
 $3\,C_{2}(\INC)$};

\node[roottoggle] (T) at (1.7,0)
{$C_{2}(X_{01})$\\
 strict root toggle};

\node[merge] (MRi) at (4.3,0)
{$M_{\mathrm{root}}^{\dagger}$};

\node[merge] (M1i) at (7.1, 1.6)
{$M_{1}^{\dagger}$};

\node[merge] (M2i) at (7.1,-1.6)
{$M_{2}^{\dagger}$};

\node[leaf] (L1i) at (10.0, 2.4)
{$L_{1}^{\dagger}$};

\node[leaf] (L2i) at (10.0, 0.8)
{$L_{2}^{\dagger}$};

\node[leaf] (L3i) at (10.0,-0.8)
{$L_{3}^{\dagger}$};

\node[leaf] (L4i) at (10.0,-2.4)
{$L_{4}^{\dagger}$};

\draw[flow] (L1.east) -- (M1.west);
\draw[flow] (L2.east) -- (M1.west);

\draw[flow] (L3.east) -- (M2.west);
\draw[flow] (L4.east) -- (M2.west);

\draw[flow] (M1.east) -- (MR.west);
\draw[flow] (M2.east) -- (MR.west);

\draw[flow] (MR.east) -- (T.west);

\draw[flow] (T.east) -- (MRi.west);

\draw[flow] (MRi.east) -- (M1i.west);
\draw[flow] (MRi.east) -- (M2i.west);

\draw[flow] (M1i.east) -- (L1i.west);
\draw[flow] (M1i.east) -- (L2i.west);

\draw[flow] (M2i.east) -- (L3i.west);
\draw[flow] (M2i.east) -- (L4i.west);

\node[anclabel] at (-3.8, 2.45)
{$\ket{0}_{a_{0}}$};

\node[anclabel] at (-3.8,-2.45)
{$\ket{0}_{a_{1}}$};

\node[anclabel] at (-0.9,0.92)
{$\ket{0}_{a_{2}}$};

\node[font=\small] at (-6.8,3.15)
{leaf predicates};

\node[font=\small] at (-2.35,2.95)
{forward predicate propagation};

\node[font=\small] at (7.1,2.95)
{inverse restoration};

\node[font=\small] at (1.7,-0.95)
{target flips iff root marker is $\ket{2}$};

\node[font=\large] at (-6.8,-3.30) {$\vdots$};
\node[font=\large] at (10.0,-3.30) {$\vdots$};

\end{tikzpicture}
}
\caption{
Balanced qubit-qutrit predicate tree.
Each leaf block $L_{j}$ uses a ternary Clifford SUM gate and one
state-selective $C_{2}(\mathrm{INC})$ operation to generate a
temporary subtree marker.
Each internal merge $M_{j}$ uses one branch-local clean ancillary
qutrit and three state-selective increments to combine two child
markers with the corresponding middle control.
After the complete all-controls-satisfied predicate reaches the root,
the unique strict toggle $C_{2}(X_{01})$ acts on the target.
The inverse tree restores all temporary marker states and returns
every ancillary qutrit to $\ket{0}$.
}
\label{fig}

\end{figure*}

\paragraph{Correctness.}

For every control subtree
$S$,
define
\begin{equation}\nonumber
\chi(S)
=
\bigwedge_{j\in S}c_j.
\label{eq:subtree-predicate}
\end{equation}
The quantity
$\chi(S)$
equals one if and only if all logical controls contained in the
subtree
$S$
are active.

The forward tree is denoted by
$W_n$.
Its purpose is to compute the root marker.
The root marker reaches
$\ket{2}$
if and only if
$\chi(S)=1$
for the full control set.
The strict toggle then flips the binary target on this branch.
Finally,
$W_n^\dagger$
reverses the temporary computation.

\begin{theorem}
For
$n=2^h-1$
controls, the circuit
\begin{equation}\nonumber
U_{\MCT_n}
=
W_n^{\dagger}
\,
\Cstrict
\,
W_n
\label{eq:factorization}
\end{equation}
implements
\begin{equation}\nonumber
\begin{split}
&
\ket{c_1\cdots c_n}
\ket{t}
\ket{\bm{0}}_{\mathrm{anc}}
\\
&\quad\longmapsto
\ket{c_1\cdots c_n}
\ket{
t\oplus
(c_1\land\cdots\land c_n)
}
\ket{\bm{0}}_{\mathrm{anc}}
\end{split}
\label{eq:mct-action}
\end{equation}
on binary-subspace inputs.
\end{theorem}

\begin{proof}
We prove by induction on the tree height that the marker of every
subtree
$S$
satisfies
\begin{equation}\nonumber
m(S)=2
\quad\Longleftrightarrow\quad
\chi(S)=1.
\label{eq:ms}
\end{equation}

For a three-control leaf, the claim follows directly from
Eq.~(\ref{eq_leaf}).
The output marker of the leaf reaches
$\ket{2}$
if and only if all three controls in that leaf are active.

Now assume that the statement holds for the two child subtrees of an
internal node.
By the induction hypothesis, the two child markers reach
$\ket{2}$
if and only if all controls in their corresponding subtrees are
active.
According to
Eq.~(\ref{eq_marker}),
the parent marker reaches
$\ket{2}$
if and only if both child markers equal
$\ket{2}$
and the middle control equals one.
This condition is exactly the conjunction of all logical controls
contained in the parent subtree.

By induction, the marker at the root of the tree equals
$\ket{2}$
if and only if every logical control equals one.
The gate
$\Cstrict$
therefore flips the binary target exactly on the all-ones branch.
Finally,
$W_n^\dagger$
reverses the forward computation, restoring every temporary marker and
returning each ancillary qutrit to
$\ket{0}$.
\end{proof}

The Supplemental Material~\cite{SupplementalMaterial} verifies the
local action of the leaf and merge primitives explicitly through
truth tables, providing a direct basis-state interpretation of the
temporary
$\ket{2}$
markers used in the construction.

\paragraph{Exact Clifford+$P_9$ compilation.}
We now evaluate the exact resource cost of the proposed decomposition. The analysis is carried out in the generic ternary Clifford+$P_9$ framework of~\cite{Bocharov2017}. The total $P_9$ count is therefore determined entirely by the number of state-selective controlled increments and by the single strict toggle applied to the target at the root of the tree.

The balanced tree contains
\begin{equation}\nonumber
L(n)
=
\frac{n+1}{4}
\label{eq:leaf-count}
\end{equation}
leaf blocks and
\begin{equation}
I(n)
=
\frac{n-3}{4}
\label{eq:merge-count}
\end{equation}
internal merges.

The forward tree therefore contains
\begin{align}\nonumber
N_{{\Csel}}^{\mathrm{forward}}(n)
&=
L(n)+3I(n)
\\
&=
n-2\nonumber
\label{eq:forward-count}
\end{align}
state-selective gates.
The inverse tree contributes the same number.

\begin{proposition}
For
$n=2^h-1$
controls, the branch-local tree requires
\begin{equation}\nonumber
\boxed{
N_{P_9}^{\mathrm{tree}}(n)
=
6n+3
}
\label{eq:tree-count}
\end{equation}
logical
$P_9$
injections and
\begin{equation}\nonumber
\boxed{
N_{\mathrm{anc}}^{\mathrm{tree}}(n)
=
\frac{n-3}{4}
}
\label{eq:tree-ancilla}
\end{equation}
clean ancillary qutrits.
\end{proposition}

\begin{proof}
The complete circuit contains
$2n-4$
state-selective gates, each costing
$3P_9$,
and one strict root toggle costing
$15P_9$.
SUM gates are Clifford.
Therefore,
\begin{align}\nonumber
N_{P_9}^{\mathrm{tree}}(n)
&=
3(2n-4)+15
\\
&=
6n+3.\nonumber
\end{align}
Each internal merge uses one ancillary qutrit, hence, we need a total number of ancillary qutrits as per Eq. (\ref{eq:merge-count}) for the decomposition.

A step-by-step derivation of the leaf count, internal-merge count, and
total number of state-selective operations is given in the
Supplemental Material~\cite{SupplementalMaterial}.

\end{proof}

The balanced-tree structure allows several operations to be executed in
parallel.
At a fixed level of the tree, nodes belonging to different branches
act on separate sets of qutrits.
Since these nodes do not share any control, marker, or ancillary
qutrits, their corresponding operations can be scheduled
simultaneously.

The circuit is therefore evaluated level by level.
All nodes at the lowest level are processed in parallel, followed by
the nodes at the next level, and so on until the root marker is
obtained.
After the target toggle, the inverse traversal proceeds in the reverse
order with the same level-wise parallelism. The depth therefore scales as

\begin{equation}\nonumber
D_{\mathrm{tree}}(n)
=
\Theta(\log_2 n).
\label{eq:tree-depth}
\end{equation}

The Supplemental Material~\cite{SupplementalMaterial} further derives an explicit depth upper bound.

\paragraph{Constructive width cost frontier.}

The construction described above assigns one clean
ancillary qutrit to each internal merge.
These ancillas remain available until the global inverse traversal is
performed.
This strategy minimizes the
$P_9$
count and allows independent branches of the tree to be evaluated in
parallel.

A different implementation strategy is possible when fewer clean
ancillary qutrits are available.
Instead of retaining all merge ancillas until the end of the forward
traversal, some of them can be cleaned and reused at later stages of
the computation.
This reduces the required ancillary workspace, but introduces
additional state-selective operations.

A lower-width schedule may reuse merge ancillas before the global
inverse traversal.
Let
\begin{equation}\nonumber
1
\leq
A
\leq
I(n)
=
\frac{n-3}{4}
\label{eq:A-range}
\end{equation}
be the maximum number of simultaneously live merge ancillas.

Each internal node whose ancilla is reused requires two immediate
cleanup operations in the forward traversal and two corresponding
operations in the inverse traversal.
Each removed retained ancilla therefore adds
$12P_9$.

\begin{proposition}
Within this predicate-tree family, the achievable
width cost frontier is
\begin{equation}
\boxed{
N_{P_9}^{\mathrm{tree}}(n,A)
=
9n-6-12A.
}
\label{eq:reuse-count}
\end{equation}
\end{proposition}

\begin{proof}
The branch-local implementation uses
$I(n)$
retained ancillary qutrits and has cost
$6n+3$.
If only
$A$
ancillas are retained simultaneously, then
$I(n)-A$
ancillas must be cleaned and reused.
Each reused ancilla adds
$12P_9$.
Therefore,
\begin{align}\nonumber
N_{P_9}^{\mathrm{tree}}(n,A)
&=
6n+3
+
12\left[I(n)-A\right]
\nonumber\ \\
&=
6n+3
+
12\left[
\frac{n-3}{4}-A
\right]
\nonumber\  \\
&=
9n-6-12A.\nonumber
\end{align}

The Supplemental Material~\cite{SupplementalMaterial} gives the corresponding cleanup schedule explicitly and shows why each removed retained ancilla adds four state-selective gates, or $12P_9$.
\end{proof}

Equation~(\ref{eq:reuse-count}) describes an explicit tradeoff between
ancillary workspace and non-Clifford cost.
Using more simultaneously live ancillary qutrits reduces the number of
cleanup operations and lowers the
$P_9$
count.
Using fewer ancillary qutrits reduces the required workspace but
increases the number of non-Clifford operations. At the low-width endpoint,
\begin{equation}\nonumber
A=1,
\end{equation}
so
\begin{equation}
\boxed{
N_{P_9}^{\mathrm{reuse}}(n)
=
9n-18.
}
\label{eq:}
\end{equation}

A depth-first one-ancilla schedule has
\begin{equation}\nonumber
D_{\mathrm{reuse}}(n)
=
\Theta(n).
\label{eq:reuse-depth}
\end{equation}
More generally, batching merge operations with
$A$
available ancillas give
\begin{equation}\nonumber
D(n,A)
=
O
\left(
\log n+\frac{n}{A}
\right).
\label{eq:depth-frontier}
\end{equation}

The frontier is constructive and model-specific.
It is not a claim of global optimality over all possible
Clifford+$P_9$ synthesis.

\paragraph{Generic control widths.}

The balanced construction gives simple closed-form resource counts and
maximal parallelism when the number of controls satisfies
$n=2^h-1$.
However, practical circuits may contain an arbitrary number of
controls.
We therefore extend the construction to general values of
$n$.
For arbitrary
$n$,
choose a balanced core width
\begin{equation}\nonumber
m
=
2^h-1
\leq n.
\label{eq:core-width}
\end{equation}
The balanced tree is first used to compute a root marker for these
$m$
controls.
The remaining
$n-m$
binary controls are then incorporated sequentially.
Each appended control requires one forward
$\Csel$
gate to update the marker and one corresponding inverse
$\Csel$
gate during uncomputation.
The appended-control identity and its inverse are verified in the
Supplemental Material.

Each appended control adds
$6P_9$.
Consequently,
\begin{equation}\nonumber
\boxed{
N_{P_9}^{\mathrm{tree}}(n;m)
=
6m+3+6(n-m)
=
6n+3.
}
\label{eq:generic-count}
\end{equation}

The ancillary requirement is
\begin{equation}\nonumber
N_{\mathrm{anc}}(n;m)
=
\frac{m-3}{4},
\label{eq:generic-ancilla}
\end{equation}
and, for this sequential extension, the depth is
\begin{equation}\nonumber
D(n;m)
=
O
\left(
\log m+n-m
\right).
\label{eq:generic-depth}
\end{equation}
The first term is the depth of the balanced core, while the second term
comes from the sequential incorporation and subsequent uncomputation of the
remaining $n-m$ controls. If $m$ is chosen as the largest balanced width not
exceeding $n$,
\begin{equation}\nonumber
m
=
2^{\lfloor \log_2(n+1)\rfloor}-1,
\end{equation}
then $0\leq n-m\leq m$. Hence the appended part can contain a linear
fraction of all controls. For example, when
\begin{equation}\nonumber
n=2^{h+1}-2,
\qquad
m=2^h-1,
\end{equation}
one has $n-m=m=\Theta(n)$, and therefore
\begin{equation}\nonumber
D(n;m)=\Theta(n).
\end{equation}
Thus, the logarithmic-depth scaling is guaranteed for the balanced family
$n=2^h-1$ (and more generally whenever $n-m=O(\log n)$), whereas this
particular arbitrary-width extension has linear depth in the worst case over
$n$.

\subsection{Comparative Study}

We now compare the proposed tree decomposition with two exact
recursive baselines derived from~\cite{Bocharov2017}.
Both baselines incorporate additional controls sequentially.
The proposed construction instead evaluates independent
groups of controls in parallel and combines their results level by
level through a balanced tree.

\paragraph{Baseline A.}
Start from the one-clean-ancilla
$12P_9$
Toffoli of \cite{Bocharov2017} repeatedly:
\begin{equation}\nonumber
N_{P_9}^{(A)}(n)
=
6n,
\qquad
N_{\mathrm{anc}}^{(A)}(n)
=
n-1.
\label{eq:baseline-a}
\end{equation}

\paragraph{Baseline B.}
Start from the ancilla-free
$15P_9$
Toffoli of \cite{Bocharov2017} repeatedly:
\begin{equation}\nonumber
N_{P_9}^{(B)}(n)
=
6n+3,
\qquad
N_{\mathrm{anc}}^{(B)}(n)
=
n-2.
\label{eq:baseline-b}
\end{equation}

Both recursive baselines have sequentially nested depth:
\begin{equation}\nonumber
D^{(A)}(n)
=
D^{(B)}(n)
=
\Theta(n).
\label{eq:baseline-depth}
\end{equation}

However, the proposed construction reduces the clean ancillary requirement and, at the same time, it replaces the linear nesting depth of the recursive baseline with logarithmic depth. For balanced widths, our tree construction matches Baseline~B in
exact
$P_9$
count while reducing the clean ancillary count from
$n-2$
to
$(n-3)/4$
and replacing linear nesting depth with logarithmic tree depth. Thus, the proposed tree construction uses asymptotically one quarter
as many clean ancillary qutrits as Baseline~B while retaining the same
exact
$P_9$
count,
\begin{equation}\nonumber
\lim_{n\rightarrow\infty}
\frac{
N_{\mathrm{anc}}^{\mathrm{tree}}(n)
}{
N_{\mathrm{anc}}^{(B)}(n)
}
=
\frac{1}{4}.
\label{eq:ancilla-ratio}
\end{equation}
Table~\ref{tab:comparison}
summarizes the resource requirements of the two recursive baselines
and the proposed tree-based constructions. The derivations of Baseline~A and Baseline~B are provided in the Supplemental Material~\cite{SupplementalMaterial}. For completeness, the Supplemental Material additionally compares the proposed ternary construction with representative binary-only fault-tolerant construction under explicitly stated conditional assumptions.
\begin{table}[!htb]
\caption{
Exact Clifford+$P_9$ comparison for balanced
$n=2^h-1$
control widths.
}
\label{tab:comparison}
\resizebox{\columnwidth}{!}{%
\begin{tabular}{lccc}
\toprule
Method
&
$P_9$ count
&
Clean ancillas
&
Depth
\\
\midrule
Baseline A \cite{Bocharov2017}
&
$6n$
&
$n-1$
&
$\Theta(n)$
\\
Baseline B \cite{Bocharov2017}
&
$6n+3$
&
$n-2$
&
$\Theta(n)$
\\
Tree, branch-local [This paper] 
&
$6n+3$
&
$(n-3)/4$
&
$\Theta(\log n)$
\\
Tree, one reused ancilla [This paper] 
&
$9n-18$
&
$1$
&
$\Theta(n)$
\\
\bottomrule
\end{tabular}
}
\end{table}

\section{Conclusion}

We introduced an intermediate-qutrit decomposition of
multi-controlled Toffoli gates with binary-subspace inputs and outputs in the ternary Clifford+$P_9$  framework. The construction evaluates the all-controls-satisfied condition through a balanced tree. We exhibited that the proposed construction uses asymptotically one quarter as many clean ancillary qutrits as state-of-the-art work while replacing linear nesting depth with logarithmic depth. For arbitrary widths, the sequential extension preserves the $6n+3$ $P_9$ count but has depth $O(\log m+n-m)$ and is therefore linear in the worst case over $n$. The same tree structure also provides a constructive tradeoff between ancillary workspace and non-Clifford cost. The present work suggests several directions for further study.
First, the predicate-tree construction can be evaluated under more
detailed fault-tolerant cost models that include magic-state
distillation, error-correction cycles, routing overhead, and
architecture-dependent connectivity constraints.
Second, alternative intermediate-qutrit setting may further reduce the $P_9$ count or ancillary requirement. Third, the proposed decomposition can be incorporated into larger
reversible subroutines, including arithmetic circuits, oracle
constructions, and amplitude-amplification algorithms, to quantify its impact at the algorithmic level. Finally, extending the same principles to higher-dimensional qudit architectures may reveal additional tradeoffs among non-Clifford cost, ancillary workspace, and circuit depth. We conclude that the proposed decomposition provides a resource-efficient building block for quantum algorithms.

\bibliographystyle{apsrev4-1}
\bibliography{ref}

\clearpage
\onecolumngrid
\begin{center}
{\large\bfseries Supplemental Material for\\[0.4em]
``Efficient Synthesis of Multi-Controlled Toffoli Gates with Ternary Clifford$+P_9$ Gates''\par}
\vspace{0.8em}
Amit Saha and Francesco Arzani\\[0.25em]
{\small DI-ENS, École Normale Supérieure, Université PSL, CNRS, INRIA,\\
45 rue d'Ulm, 75005 Paris, France}
\end{center}
\vspace{0.8em}
\twocolumngrid

\section{Example decomposition circuit}

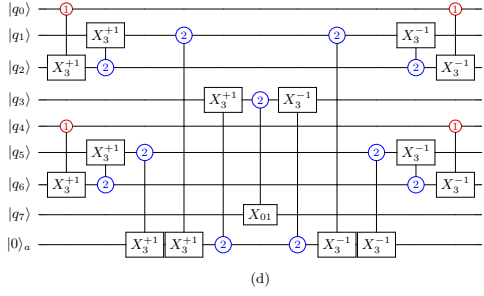
\begin{figure}[!htb]
	\centering
	\[
	\scalebox{0.6}{
		\Qcircuit @R=0.5em @C=0.1em {
			\lstick{\ket{q_{0}}} & \qw & \qw & \qw & \qw & \qw & \onecontrol & \qw & \qw & \qw & \qw & \qw & \qw  & \qw & \qw  & \qw & \onecontrol & \qw & \qw & \qw & \qw & \qw\\
			\lstick{\ket{q_{1}}} & \qw & \qw & \qw & \qw & \qw & \qw \qwx & \gate{X_3^{+1}} & \qw &  \twocontrol & \qw &  \qw  & \qw & \twocontrol & \qw & \gate{X_3^{-1}} & \qw \qwx & \qw & \qw & \qw & \qw & \qw\\
			\lstick{\ket{q_{2}}} & \qw & \qw & \qw & \qw & \qw & \gate{X_3^{+1}} \qwx & \twocontrol \qwx & \qw &  \qw \qwx & \qw  & \qw  & \qw & \qw \qwx & \qw  & \twocontrol \qwx & \gate{X_3^{-1}} \qwx & \qw & \qw & \qw & \qw & \qw\\
			\lstick{\ket{q_{3}}} & \qw & \qw & \qw & \qw & \qw & \qw & \qw & \qw  & \qw \qwx & \gate{X_3^{+1}}  & \twocontrol  & \gate{X_3^{-1}} & \qw \qwx & \qw & \qw & \qw & \qw & \qw & \qw & \qw & \qw\\
			\lstick{\ket{q_{4}}} & \qw & \qw & \qw & \qw & \qw & \onecontrol & \qw & \qw & \qw \qwx & \qw \qwx  & \qw \qwx  & \qw \qwx & \qw \qwx & \qw  & \qw & \onecontrol & \qw & \qw & \qw & \qw & \qw\\
			\lstick{\ket{q_{5}}} & \qw & \qw & \qw & \qw & \qw & \qw \qwx & \gate{X_3^{+1}} & \twocontrol & \qw \qwx & \qw \qwx  & \qw \qwx  & \qw \qwx &\qw \qwx & \twocontrol  & \gate{X_3^{-1}} & \qw \qwx & \qw & \qw & \qw & \qw & \qw\\
			\lstick{\ket{q_{6}}} & \qw & \qw & \qw & \qw & \qw & \gate{X_3^{+1}} \qwx & \twocontrol \qwx & \qw\qwx  & \qw\qwx & \qw\qwx  & \qw\qwx  & \qw\qwx & \qw\qwx & \qw \qwx & \twocontrol \qwx & \gate{X_3^{-1}} \qwx & \qw & \qw & \qw & \qw & \qw\\
			\lstick{\ket{q_{7}}} & \qw & \qw & \qw & \qw & \qw & \qw & \qw  & \qw\qwx & \qw\qwx & \qw\qwx &  \gate{X_{01}} \qwx & \qw\qwx & \qw\qwx & \qw\qwx  & \qw & \qw & \qw & \qw & \qw & \qw & \qw\\
			\lstick{\ket{0}_{a}} & \qw & \qw & \qw & \qw & \qw & \qw  & \qw & \gate{X_3^{+1}} \qwx & \gate{X_3^{+1}} \qwx & \twocontrol \qwx  & \qw   & \twocontrol \qwx & \gate{X_3^{-1}} \qwx & \gate{X_3^{-1}} \qwx  & \qw & \qw  & \qw & \qw & \qw & \qw & \qw\\
			&  &  &  &  &  &  &   &  &  &  &   &  &  &   &  &  &  &  &  &  & \\
			&  &  &  &  &  &  &   &  &  &  &    &  &  &   &  &  &  &  &  &  & \\
			&  &  &  &  &  &  &   &  &  &  &  \mbox{(d)}  &  &  &   &  &  &  &  &  &  & \\
	}}
	\]
	\caption{Decomposition of a seven-controlled MCT gate. The logical controls are $q_{0},\ldots,q_{6}$ and the target is $q_{7}$. Red circles denote activation on $\ket{1}$, whereas blue circles denote activation on $\ket{2}$. The ternary Clifford SUM and $\mathrm{SUM}^{\dagger}$ gates generate and uncompute the two leaf markers. The state-selective $C_{2}(\mathrm{INC})^{\pm1}$ operations propagate and remove the complete all-controls-satisfied predicate using the clean ancillary qutrit $a$. The root marker activates the unique strict target toggle $C_{2}(X_{01})$. The circuit requires $45P_{9}$ injections and one clean ancillary qutrit.}
	\label{fig:seven-control-balanced-tree}
\end{figure}

\begin{figure}[!htb]
\centering
\[
\scalebox{0.4}{
\Qcircuit @R=0.45em @C=0.18em {
\lstick{\ket{q_{0}}} & \onecontrol & \qw & \qw & \qw & \qw & \qw & \qw & \qw & \qw & \qw & \qw & \qw & \qw & \qw & \qw & \qw & \qw & \qw & \qw & \qw & \qw & \qw & \onecontrol & \qw \\
\lstick{\ket{q_{1}}} & \gate{X_3^{+1}}\qwx[-1] & \twocontrol & \qw & \qw & \qw & \qw & \qw & \qw & \qw & \qw & \qw & \qw & \qw & \qw & \qw & \qw & \qw & \qw & \qw & \qw & \qw & \twocontrol & \gate{X_3^{-1}}\qwx[-1] & \qw \\
\lstick{\ket{q_{2}}} & \qw & \gate{X_3^{+1}}\qwx[-1] & \twocontrol & \qw & \qw & \qw & \qw & \qw & \qw & \qw & \qw & \qw & \qw & \qw & \qw & \qw & \qw & \qw & \qw & \qw & \twocontrol & \gate{X_3^{-1}}\qwx[-1] & \qw & \qw \\
\lstick{\ket{q_{3}}} & \qw & \qw & \qw & \qw & \qw & \qw & \gate{X_3^{+1}}\qwx[13] & \qw & \twocontrol & \qw & \qw & \qw & \qw & \qw & \twocontrol & \qw & \gate{X_3^{-1}}\qwx[13] & \qw & \qw & \qw & \qw & \qw & \qw & \qw \\
\lstick{\ket{q_{4}}} & \onecontrol & \qw & \qw & \qw & \qw & \qw & \qw & \qw & \qw & \qw & \qw & \qw & \qw & \qw & \qw & \qw & \qw & \qw & \qw & \qw & \qw & \qw & \onecontrol & \qw \\
\lstick{\ket{q_{5}}} & \gate{X_3^{+1}}\qwx[-1] & \twocontrol & \qw & \qw & \qw & \qw & \qw & \qw & \qw & \qw & \qw & \qw & \qw & \qw & \qw & \qw & \qw & \qw & \qw & \qw & \qw & \twocontrol & \gate{X_3^{-1}}\qwx[-1] & \qw \\
\lstick{\ket{q_{6}}} & \qw & \gate{X_3^{+1}}\qwx[-1] & \qw & \qw & \twocontrol & \qw & \qw & \qw & \qw & \qw & \qw & \qw & \qw & \qw & \qw & \qw & \qw & \qw & \twocontrol & \qw & \qw & \gate{X_3^{-1}}\qwx[-1] & \qw & \qw \\
\lstick{\ket{q_{7}}} & \qw & \qw & \qw & \qw & \qw & \qw & \qw & \qw & \qw & \qw & \gate{X_3^{+1}}\qwx[11] & \twocontrol & \gate{X_3^{-1}}\qwx[11] & \qw & \qw & \qw & \qw & \qw & \qw & \qw & \qw & \qw & \qw & \qw \\
\lstick{\ket{q_{8}}} & \onecontrol & \qw & \qw & \qw & \qw & \qw & \qw & \qw & \qw & \qw & \qw & \qw & \qw & \qw & \qw & \qw & \qw & \qw & \qw & \qw & \qw & \qw & \onecontrol & \qw \\
\lstick{\ket{q_{9}}} & \gate{X_3^{+1}}\qwx[-1] & \twocontrol & \qw & \qw & \qw & \qw & \qw & \qw & \qw & \qw & \qw & \qw & \qw & \qw & \qw & \qw & \qw & \qw & \qw & \qw & \qw & \twocontrol & \gate{X_3^{-1}}\qwx[-1] & \qw \\
\lstick{\ket{q_{10}}} & \qw & \gate{X_3^{+1}}\qwx[-1] & \qw & \twocontrol & \qw & \qw & \qw & \qw & \qw & \qw & \qw & \qw & \qw & \qw & \qw & \qw & \qw & \qw & \qw & \twocontrol & \qw & \gate{X_3^{-1}}\qwx[-1] & \qw & \qw \\
\lstick{\ket{q_{11}}} & \qw & \qw & \qw & \qw & \qw & \qw & \qw & \gate{X_3^{+1}}\qwx[6] & \qw & \twocontrol & \qw & \qw & \qw & \twocontrol & \qw & \gate{X_3^{-1}}\qwx[6] & \qw & \qw & \qw & \qw & \qw & \qw & \qw & \qw \\
\lstick{\ket{q_{12}}} & \onecontrol & \qw & \qw & \qw & \qw & \qw & \qw & \qw & \qw & \qw & \qw & \qw & \qw & \qw & \qw & \qw & \qw & \qw & \qw & \qw & \qw & \qw & \onecontrol & \qw \\
\lstick{\ket{q_{13}}} & \gate{X_3^{+1}}\qwx[-1] & \twocontrol & \qw & \qw & \qw & \qw & \qw & \qw & \qw & \qw & \qw & \qw & \qw & \qw & \qw & \qw & \qw & \qw & \qw & \qw & \qw & \twocontrol & \gate{X_3^{-1}}\qwx[-1] & \qw \\
\lstick{\ket{q_{14}}} & \qw & \gate{X_3^{+1}}\qwx[-1] & \qw & \qw & \qw & \twocontrol & \qw & \qw & \qw & \qw & \qw & \qw & \qw & \qw & \qw & \qw & \qw & \twocontrol & \qw & \qw & \qw & \gate{X_3^{-1}}\qwx[-1] & \qw & \qw \\
\lstick{\ket{q_{15}}} & \qw & \qw & \qw & \qw & \qw & \qw & \qw & \qw & \qw & \qw & \qw & \gate{X_{01}}\qwx[-8] & \qw & \qw & \qw & \qw & \qw & \qw & \qw & \qw & \qw & \qw & \qw & \qw \\
\lstick{\ket{0}_{a_{0}}} & \qw & \qw & \gate{X_3^{+1}}\qwx[-14] & \qw & \gate{X_3^{+1}}\qwx[-10] & \qw & \twocontrol & \qw & \qw & \qw & \qw & \qw & \qw & \qw & \qw & \qw & \twocontrol & \qw & \gate{X_3^{-1}}\qwx[-10] & \qw & \gate{X_3^{-1}}\qwx[-14] & \qw & \qw & \qw \\
\lstick{\ket{0}_{a_{1}}} & \qw & \qw & \qw & \gate{X_3^{+1}}\qwx[-7] & \qw & \gate{X_3^{+1}}\qwx[-3] & \qw & \twocontrol & \qw & \qw & \qw & \qw & \qw & \qw & \qw & \twocontrol & \qw & \gate{X_3^{-1}}\qwx[-3] & \qw & \gate{X_3^{-1}}\qwx[-7] & \qw & \qw & \qw & \qw \\
\lstick{\ket{0}_{a_{2}}} & \qw & \qw & \qw & \qw & \qw & \qw & \qw & \qw & \gate{X_3^{+1}}\qwx[-15] & \gate{X_3^{+1}}\qwx[-7] & \twocontrol & \qw & \twocontrol & \gate{X_3^{-1}}\qwx[-7] & \gate{X_3^{-1}}\qwx[-15] & \qw & \qw & \qw & \qw & \qw & \qw & \qw & \qw & \qw
}
}
\]
\caption{
Fifteen-control proposed intermediate-qutrit decomposition of an MCT gate.
The logical controls are $q_{0},\ldots,q_{14}$ and the target is
$q_{15}$.
A red circle labeled $1$ denotes activation on $\ket{1}$, whereas
a blue circle labeled $2$ denotes activation on $\ket{2}$.
The branch-local propagation operations are horizontally staggered to
keep the vertical connections visually distinct.
Four ternary Clifford SUM gates generate the leaf markers in parallel.
The ancillary qutrits $a_{0}$ and $a_{1}$ form two seven-control
subtree predicates, while $a_{2}$ combines them into the root
predicate.
The marker on $q_{7}$ activates the unique strict root toggle
$C_{2}(X_{01})$ on $q_{15}$.
The inverse half of the circuit restores all temporary marker states
and returns $a_{0}$, $a_{1}$, and $a_{2}$ to $\ket{0}$.
The circuit requires $93P_{9}$ injections and three clean ancillary
qutrits.
}
\label{fig:fifteen-control-balanced-tree}
\end{figure}
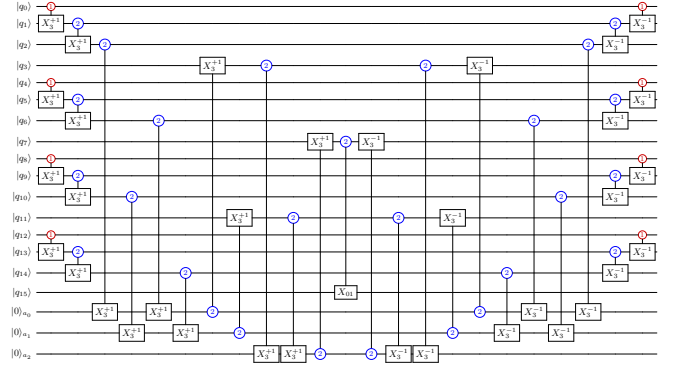

\section{Local qutrit identities}

\subsection{Three-control leaf block}

The leaf block is
\begin{equation}\nonumber
L(u,v,w)
=
\Csel_{v,w}
\,
\SUM_{u,v}.
\end{equation}

For binary inputs
$u,v,w\in\{0,1\}$,
the truth table is:

\begin{table}[h]
\caption{
Truth table of the three-control leaf block.
}
\label{tab:leaf-table}
\begin{tabular}{cccc}
\toprule
Input
&
after $\SUM_{u,v}$
&
after $\Csel_{v,w}$
&
$w_{\mathrm{out}}=2$?
\\
\midrule
$(0,0,0)$ & $(0,0,0)$ & $(0,0,0)$ & no \\
$(0,0,1)$ & $(0,0,1)$ & $(0,0,1)$ & no \\
$(0,1,0)$ & $(0,1,0)$ & $(0,1,0)$ & no \\
$(0,1,1)$ & $(0,1,1)$ & $(0,1,1)$ & no \\
$(1,0,0)$ & $(1,1,0)$ & $(1,1,0)$ & no \\
$(1,0,1)$ & $(1,1,1)$ & $(1,1,1)$ & no \\
$(1,1,0)$ & $(1,2,0)$ & $(1,2,1)$ & no \\
$(1,1,1)$ & $(1,2,1)$ & $(1,2,2)$ & yes \\
\bottomrule
\end{tabular}
\end{table}

Thus,
\begin{equation}\nonumber
w_{\mathrm{out}}=2
\quad\Longleftrightarrow\quad
u=v=w=1.
\end{equation}

\subsection{Internal merge}

Let
$x$
and
$y$
be the markers of two previously computed child subtrees,
let
$m$
be a binary middle control,
and let
$a$
be a clean ancillary qutrit.
The merge is
\begin{equation}\nonumber
M(x,m,y;a)
=
\Csel_{a,m}
\,
\Csel_{y,a}
\,
\Csel_{x,a}.
\end{equation}

Define
\begin{equation}\nonumber
b_x
=
\delta_{x,2},
\qquad
b_y
=
\delta_{y,2}.
\end{equation}
Starting from
$a=0$,
the first two chronological operations give
\begin{equation}\nonumber
a=b_x+b_y\pmod3.
\end{equation}
Hence,
\begin{equation}\nonumber
a=2
\quad\Longleftrightarrow\quad
b_x=b_y=1.
\end{equation}

\begin{table}[h]
\caption{
Truth table of the internal merge.
The middle control is initially binary.
}
\label{tab:merge-table}
\begin{tabular}{cccccc}
\toprule
$b_x$
&
$b_y$
&
$m_{\mathrm{in}}$
&
$a$
&
$m_{\mathrm{out}}$
&
$m_{\mathrm{out}}=2$?
\\
\midrule
$0$ & $0$ & $0$ & $0$ & $0$ & no \\
$0$ & $0$ & $1$ & $0$ & $1$ & no \\
$0$ & $1$ & $0$ & $1$ & $0$ & no \\
$0$ & $1$ & $1$ & $1$ & $1$ & no \\
$1$ & $0$ & $0$ & $1$ & $0$ & no \\
$1$ & $0$ & $1$ & $1$ & $1$ & no \\
$1$ & $1$ & $0$ & $2$ & $1$ & no \\
$1$ & $1$ & $1$ & $2$ & $2$ & yes \\
\bottomrule
\end{tabular}
\end{table}

Thus,
\begin{equation}\nonumber
m_{\mathrm{out}}=2
\quad\Longleftrightarrow\quad
x=y=2
\ \text{and}\
m_{\mathrm{in}}=1.
\end{equation}

\section{Balanced-tree gate count}

For
\begin{equation}\nonumber
n=2^h-1,
\qquad
h\geq2,
\end{equation}
the number of leaf blocks is
\begin{equation}\nonumber
L(n)
=
2^{h-2}
=
\frac{n+1}{4}.
\end{equation}

The number of internal merges is
\begin{equation}\nonumber
I(n)
=
L(n)-1
=
\frac{n-3}{4}.
\end{equation}

Each leaf contributes one forward
$\Csel$
gate.
Each internal merge contributes three.
Therefore,
\begin{align}
N_{\mathrm{sel}}^{\mathrm{forward}}(n)
&=
L(n)+3I(n) \nonumber
\\ 
&=
\frac{n+1}{4}
+
3\frac{n-3}{4} \nonumber
\\ 
&=
n-2. \nonumber
\end{align}

The inverse tree contributes the same number:
\begin{equation}\nonumber
N_{\mathrm{sel}}(n)
=
2n-4.
\end{equation}

Using
\begin{equation}\nonumber
N_{P_9}
\left[
\Csel^{\pm1}
\right]
=
3,
\end{equation}
and
\begin{equation}\nonumber
N_{P_9}
\left[
\Cstrict
\right]
=
15,
\end{equation}
we obtain
\begin{align}
N_{P_9}^{\mathrm{tree}}(n)
&=
3(2n-4)+15 \nonumber
\\
&=
6n+3. \nonumber
\end{align}

\section{Primitive-gate depth}

This section reports a depth upper bound at the level of the declared
qutrit primitives:
\begin{equation}\nonumber
\SUM,
\qquad
\Csel,
\qquad
\Cstrict.
\end{equation}
It is not an exact
$P_9$-depth
claim.
An exact
$P_9$-depth
requires an explicit schedule for the internal Clifford+$P_9$
synthesis of each macro-gate.

For a balanced tree of height
$h$,
the forward traversal consists of:

\begin{enumerate}
\item one Clifford SUM layer and one
$\Csel$
layer for the leaves;
\item
$h-2$
internal levels, each requiring three sequential
$\Csel$
layers.
\end{enumerate}

Thus,
\begin{equation}\nonumber
D_W(h)
=
2+3(h-2)
=
3h-4.
\end{equation}

Including the inverse traversal and the strict root toggle,
\begin{equation}\nonumber
\boxed{
D_{\mathrm{prim}}(h)
=
2D_W(h)+1
=
6h-7.
}
\end{equation}

Since
\begin{equation}\nonumber
h
=
\log_2(n+1),
\end{equation}
we obtain
\begin{equation}\nonumber
D_{\mathrm{prim}}(n)
=
6\log_2(n+1)-7.
\end{equation}

For seven controls,
\begin{equation}\nonumber
D_{\mathrm{prim}}(7)
=
11.
\end{equation}

For fifteen controls,
\begin{equation}\nonumber
D_{\mathrm{prim}}(15)
=
17.
\end{equation}

These are primitive-gate depth upper bounds for the branch-local
schedule.

\section{Constructive width--cost frontier}

Let
\begin{equation}\nonumber
1
\leq
A
\leq
I(n)
=
\frac{n-3}{4}
\end{equation}
be the maximum number of simultaneously live merge ancillas.

The branch-local endpoint uses
\begin{equation}\nonumber
A=I(n)
\end{equation}
and postpones every ancillary cleanup until the global inverse
traversal.

If a merge ancilla must be reused during the forward traversal, it
is returned to
$\ket{0}$
by applying
\begin{equation}\nonumber
\Csel_{y,a}^{\dagger},
\qquad
\Csel_{x,a}^{\dagger}.
\end{equation}
The inverse tree contains the corresponding additional operations.
Therefore, every removed retained ancilla adds four state-selective
macro-gates, or
\begin{equation}\nonumber
12P_9.
\end{equation}

The resulting achievable frontier is
\begin{align}
N_{P_9}^{\mathrm{tree}}(n,A)
&=
6n+3
+
12
\left[
I(n)-A
\right] \nonumber
\\
&=
6n+3
+
12
\left[
\frac{n-3}{4}-A
\right] \nonumber
\\
&=
\boxed{
9n-6-12A. \nonumber
}
\end{align}

At
\begin{equation}\nonumber
A=1,
\end{equation}
we obtain
\begin{equation}\nonumber
N_{P_9}^{\mathrm{reuse}}(n)
=
9n-18.
\end{equation}

At
\begin{equation}\nonumber
A=\frac{n-3}{4},
\end{equation}
we recover
\begin{equation}\nonumber
N_{P_9}^{\mathrm{tree}}(n)
=
6n+3.
\end{equation}

\begin{table}[h]
\caption{
Constructive width--cost frontier for a fifteen-control MCT gate.
}
\label{tab:fifteen-frontier}
\begin{tabular}{ccc}
\toprule
Live merge ancillas
&
$P_9$ count
&
Schedule
\\
\midrule
$1$ & $117$ & depth-first reuse \\
$2$ & $105$ & intermediate batching \\
$3$ & $93$ & branch-local parallelism \\
\bottomrule
\end{tabular}
\end{table}

This frontier is model-specific.
It is not a global lower bound over every possible
Clifford+$P_9$
synthesis.

\section{Extension to arbitrary control width}

Let
\begin{equation}\nonumber
m=2^h-1
\leq n
\end{equation}
be a chosen balanced-core width.
The core network produces a marker
$r$
satisfying
\begin{equation}\nonumber
r=2
\quad\Longleftrightarrow\quad
c_1=\cdots=c_m=1.
\end{equation}

Let
$c$
be an additional binary control.
Apply
\begin{equation}\nonumber
\Csel_{r,c}.
\end{equation}
Its action is
\begin{equation}\nonumber
\ket{r,c}
\longmapsto
\ket{
r,
c+\delta_{r,2}\!\!\!\pmod3
}.
\end{equation}

Because
$c\in\{0,1\}$,
the new marker satisfies
\begin{equation}\nonumber
c_{\mathrm{out}}=2
\quad\Longleftrightarrow\quad
r=2
\ \text{and}\
c_{\mathrm{in}}=1.
\end{equation}

Thus, one appended control extends the predicate by one Boolean
conjunction.
The matching inverse gate restores the appended control during
uncomputation.

Appending
$n-m$
controls requires
\begin{equation}\nonumber
n-m
\end{equation}
forward gates and
\begin{equation}\nonumber
n-m
\end{equation}
inverse gates.
Each appended control therefore contributes
\begin{equation}\nonumber
6P_9.
\end{equation}

Hence,
\begin{align}
N_{P_9}^{\mathrm{tree}}(n;m)
&=
6m+3
+
6(n-m) \nonumber
\\
&=
\boxed{
6n+3. \nonumber
}
\end{align}

The ancillary requirement is
\begin{equation}\nonumber
N_{\mathrm{anc}}(n;m)
=
\frac{m-3}{4},
\end{equation}
and a primitive-gate depth upper bound is
\begin{equation}\nonumber
D_{\mathrm{prim}}(n;m)
\leq
6\log_2(m+1)-7
+
2(n-m).
\end{equation}

Choosing the largest balanced core maximizes parallelism.
Choosing a smaller core reduces ancillary width but increases the
sequential extension tail.

\section{Recursive ternary baselines derived from \cite{Bocharov2017}}

The labels Baseline A and Baseline B are introduced for convenience.
They are not terminology used in
\cite{Bocharov2017}.

\subsection{Baseline A}

The work of
\cite{Bocharov2017}
gives a two-control Toffoli emulation with:

\begin{equation}\nonumber
N_{P_9}^{(A)}(2)=12,
\qquad
N_{\mathrm{anc}}^{(A)}(2)=1.
\end{equation}

Each control extension adds:

\begin{equation}\nonumber
6P_9
\end{equation}
and one clean ancillary qutrit.
Therefore,
\begin{equation}\nonumber
N_{P_9}^{(A)}(n)
=
6n,
\end{equation}
and
\begin{equation}\nonumber
N_{\mathrm{anc}}^{(A)}(n)
=
n-1.
\end{equation}

\subsection{Baseline B}

The work of
\cite{Bocharov2017}
also gives an ancilla-free two-control Toffoli emulation with:

\begin{equation}\nonumber
N_{P_9}^{(B)}(2)=15,
\qquad
N_{\mathrm{anc}}^{(B)}(2)=0.
\end{equation}

Each control extension again adds:

\begin{equation}\nonumber
6P_9
\end{equation}
and one clean ancillary qutrit.
Therefore,
\begin{equation}\nonumber
N_{P_9}^{(B)}(n)
=
6n+3,
\end{equation}
and
\begin{equation}\nonumber
N_{\mathrm{anc}}^{(B)}(n)
=
n-2.
\end{equation}

Repeated control extensions are sequentially nested, so both
baseline depths scale as
\begin{equation}\nonumber
\Theta(n).
\end{equation}

\section{
Conditional Steane-type comparison with recursive ternary baselines
}
\label{sec:bocharov-steane-comparison}

The logical comparison in the main text is performed within the
generic ternary Clifford+$P_9$ model of
\cite{Bocharov2017}.
To examine the possible encoded-level implications of the proposed
tree, we now compare it with the two recursive ternary baselines
defined in the preceding section under a common seven-qutrit
Steane-type code stack.

This comparison is conditional.
\cite{Bocharov2017} specifies the logical ternary
Clifford+$P_9$ resource model, but it does not fix a unique
seven-qutrit Steane-type implementation, syndrome-extraction circuit,
routing architecture, or logical $P_9$ factory.
Accordingly, the analysis below separates the exact logical
$P_9$ contribution from architecture-dependent encoded overheads.

\subsection{Logical resources}

For balanced widths
\begin{equation}\nonumber
n=2^h-1,
\end{equation}
the branch-local tree requires
\begin{equation}\nonumber
N_{P_9}^{\mathrm{tree}}(n)
=
6n+3,
\end{equation}
\begin{equation}\nonumber
N_{\mathrm{anc}}^{\mathrm{tree}}(n)
=
\frac{n-3}{4},
\end{equation}
and primitive-gate depth
\begin{equation}\nonumber
D_{\mathrm{tree}}(n)
=
6\log_2(n+1)-7.
\end{equation}

The recursive ternary baselines derived from
\cite{Bocharov2017}
have
\begin{align}
N_{P_9}^{(A)}(n)
&=
6n,
&
N_{\mathrm{anc}}^{(A)}(n)
&=
n-1, \nonumber
\\
N_{P_9}^{(B)}(n)
&=
6n+3,
&
N_{\mathrm{anc}}^{(B)}(n)
&=
n-2. \nonumber
\end{align}

Repeated control extensions are sequentially nested.
Their scheduled depths therefore scale as
\begin{equation}\nonumber
D_A(n)
=
\Theta(n),
\qquad
D_B(n)
=
\Theta(n).
\end{equation}

\begin{table}[h]
\caption{
Logical comparison with recursive ternary baselines derived from
\cite{Bocharov2017}.
The exact primitive-gate depth of each recursively nested baseline
depends on the selected scheduling convention; its asymptotic scaling
is linear.
}
\label{tab:bocharov-logical-comparison}
\begin{tabular}{lccc}
\toprule
Method
&
$P_9$ count
&
Clean ancillas
&
Depth
\\
\midrule
Tree, branch-local
&
$6n+3$
&
$(n-3)/4$
&
$\Theta(\log n)$
\\
Baseline A
&
$6n$
&
$n-1$
&
$\Theta(n)$
\\
Baseline B
&
$6n+3$
&
$n-2$
&
$\Theta(n)$
\\
\bottomrule
\end{tabular}
\end{table}

For representative balanced widths:

\begin{table}[h]
\caption{
Representative logical-resource counts.
}
\label{tab:bocharov-representative-counts}
\resizebox{\columnwidth}{!}{%
\begin{tabular}{lcccc}
\toprule
Controls
&
Method
&
$P_9$ count
&
Clean ancillas
&
Tree depth
\\
\midrule
$7$
&
Tree, branch-local
&
$45$
&
$1$
&
$11$
\\
$7$
&
Baseline A
&
$42$
&
$6$
&
linear
\\
$7$
&
Baseline B
&
$45$
&
$5$
&
linear
\\
\midrule
$15$
&
Tree, branch-local
&
$93$
&
$3$
&
$17$
\\
$15$
&
Baseline A
&
$90$
&
$14$
&
linear
\\
$15$
&
Baseline B
&
$93$
&
$13$
&
linear
\\
\bottomrule
\end{tabular}}
\end{table}

\subsection{Common encoded-resource model}

Let
\begin{equation}\nonumber
\mathcal{F}_{P_9}
\left(
\varepsilon,
p_3,
k_3
\right)
\end{equation}
denote the encoded carrier-cycle cost of preparing and injecting one
logical
$P_9$
resource state with target failure probability
$\varepsilon$.
Here,
$p_3$
is the physical qutrit-location error probability and
$k_3$
is the number of ternary concatenation levels.

Let
\begin{equation}\nonumber
R_j(n)
\end{equation}
denote the remaining encoded overhead of construction
$j$,
where
\begin{equation}\nonumber
j\in
\left\{
\mathrm{tree},
A,
B
\right\}.
\end{equation}
The remainder term includes encoded ternary Clifford operations,
syndrome extraction, clean-ancilla preparation, storage,
routing, reset, measurement, and classical feed-forward where
required.

For an allowed failure probability
\begin{equation}\nonumber
\varepsilon_{\MCT}
\end{equation}
for one encoded MCT invocation, a conservative union-bound
allocation gives
\begin{align}
V_{\mathrm{tree}}(n)
={}&
(6n+3)
\,
\mathcal{F}_{P_9}
\left(
\frac{\varepsilon_{\MCT}}{6n+3},
p_3,
k_3
\right)
+
R_{\mathrm{tree}}(n), \nonumber
\\
V_A(n)
={}&
6n
\,
\mathcal{F}_{P_9}
\left(
\frac{\varepsilon_{\MCT}}{6n},
p_3,
k_3
\right)
+
R_A(n), \nonumber
\label{eq:encoded-baseline-a-cost}
\\
V_B(n)
={}&
(6n+3)
\,
\mathcal{F}_{P_9}
\left(
\frac{\varepsilon_{\MCT}}{6n+3},
p_3,
k_3
\right)
+
R_B(n). \nonumber
\end{align}

These expressions use the same encoded qutrit code stack and the same
logical
$P_9$
factory model for all three constructions.

\subsection{Comparison with Baseline B}

The proposed tree and Baseline B have identical logical
$P_9$
counts:
\begin{equation}\nonumber
N_{P_9}^{\mathrm{tree}}(n)
=
N_{P_9}^{(B)}(n)
=
6n+3.
\end{equation}

Consequently, their encoded
$P_9$
factory contributions cancel exactly:
\begin{equation}\nonumber
V_B(n)-V_{\mathrm{tree}}(n)
=
R_B(n)-R_{\mathrm{tree}}(n).
\label{eq:baseline-b-cancellation}
\end{equation}

The comparison is therefore determined entirely by encoded remainder
costs.
The tree reduces the clean ancillary-qutrit requirement from
\begin{equation}\nonumber
n-2
\end{equation}
to
\begin{equation}\nonumber
\frac{n-3}{4},
\end{equation}
and replaces sequentially nested depth by logarithmic depth.

Thus, under any common encoded implementation for which the reduced
ancilla preparation, storage, routing, and error-correction exposure
lower the remainder term,
\begin{equation}\nonumber
R_{\mathrm{tree}}(n)
<
R_B(n),
\end{equation}
the proposed tree has lower total encoded cost:
\begin{equation}\nonumber
\boxed{
V_{\mathrm{tree}}(n)
<
V_B(n).
}
\label{eq:baseline-b-advantage}
\end{equation}

The guaranteed logical improvement is the ancillary-width and depth
reduction.
The physical-volume improvement remains conditional on the selected
encoded architecture.

\subsection{Comparison with Baseline A}

Baseline A requires three fewer logical
$P_9$
injections:
\begin{equation}\nonumber
N_{P_9}^{\mathrm{tree}}(n)
-
N_{P_9}^{(A)}(n)
=
3.
\end{equation}

Define
\begin{align}
\Delta \mathcal{P}_A(n)
={}&
(6n+3)
\,
\mathcal{F}_{P_9}
\left(
\frac{\varepsilon_{\MCT}}{6n+3},
p_3,
k_3
\right)
\nonumber
\\
&
-
6n
\,
\mathcal{F}_{P_9}
\left(
\frac{\varepsilon_{\MCT}}{6n},
p_3,
k_3
\right). \nonumber
\label{eq:baseline-a-extra-p9-cost}
\end{align}

The proposed tree is favorable whenever the reduction in encoded
remainder cost exceeds the additional
$P_9$
factory cost:
\begin{equation}\nonumber
\boxed{
R_A(n)-R_{\mathrm{tree}}(n)
>
\Delta \mathcal{P}_A(n).
}
\label{eq:baseline-a-crossover-exact}
\end{equation}

If the per-state output-accuracy dependence of the
$P_9$
factory is neglected in a first-order approximation, then
\begin{equation}\nonumber
\Delta \mathcal{P}_A(n)
\approx
3\mathcal{F}_{P_9}.
\end{equation}
The crossover condition simplifies to
\begin{equation}\nonumber
\boxed{
R_A(n)-R_{\mathrm{tree}}(n)
>
3\mathcal{F}_{P_9}.
}
\label{eq:baseline-a-crossover-simple}
\end{equation}

Although Baseline A is slightly cheaper in logical
$P_9$
count, it uses
\begin{equation}\nonumber
n-1
\end{equation}
clean ancillary qutrits and has linear nesting depth.
The proposed tree can therefore become favorable when ancillary
preparation, storage, routing, and error-correction exposure are
included.

\subsection{Ancilla-time sensitivity model}

To make the width--depth effect explicit, introduce a simplified
encoded sensitivity model:
\begin{equation}\nonumber
R_j(n)
=
\mu_3
N_{\mathrm{anc}}^{(j)}(n)
D_j(n)
+
R_j^{\mathrm{other}}(n),
\label{eq:ancilla-time-model}
\end{equation}
where
\begin{equation}\nonumber
\mu_3>0
\end{equation}
is the encoded carrier-cycle cost of maintaining one clean logical
qutrit ancilla for one primitive-gate time step.
The quantity
\begin{equation}\nonumber
R_j^{\mathrm{other}}(n)
\end{equation}
collects all remaining overheads not included in this
ancilla-time proxy.

For the proposed tree,
\begin{equation}\nonumber
R_{\mathrm{tree}}^{\mathrm{anc}}(n)
=
\mu_3
\frac{n-3}{4}
\left[
6\log_2(n+1)-7
\right].
\label{eq:tree-ancilla-time}
\end{equation}

For the recursively nested baselines,
\begin{align}
R_A^{\mathrm{anc}}(n)
&=
\mu_3
(n-1)
D_A(n), \nonumber
\\
R_B^{\mathrm{anc}}(n)
&=
\mu_3
(n-2)
D_B(n). \nonumber
\end{align}

Because
\begin{equation}\nonumber
D_A(n),
D_B(n)
=
\Theta(n),
\end{equation}
whereas
\begin{equation}\nonumber
D_{\mathrm{tree}}(n)
=
\Theta(\log n),
\end{equation}
the ancillary-time proxy scales as
\begin{align}
R_{\mathrm{tree}}^{\mathrm{anc}}(n)
&=
O(n\log n), \nonumber
\\
R_A^{\mathrm{anc}}(n),
R_B^{\mathrm{anc}}(n)
&=
O(n^2). \nonumber
\end{align}

This model is illustrative rather than hardware independent.
It shows that the logical tree advantage can translate into a growing
encoded-level benefit whenever clean-ancilla storage and repeated
error correction contribute appreciably to the total resource volume.

\subsection{Summary of the encoded ternary comparison}

The comparison with the recursive ternary baselines leads to two
distinct conclusions. First, relative to Baseline B, the proposed tree preserves the exact
logical
$P_9$
count while reducing ancillary width and scheduled depth.
Its
$P_9$
factory contribution is therefore identical, and any reduction in
encoded remainder cost gives a direct physical advantage.

Second, relative to Baseline A, the proposed tree trades three
additional logical
$P_9$
injections for a substantial reduction in ancillary width and
scheduled depth.
Its encoded advantage is conditional on the reduction in storage,
routing, and error-correction overhead exceeding the cost of the
three additional
$P_9$
resource states.

\section{Conditional Steane-code comparison with binary-only circuits}

The logical analysis in the main text is exact within the ternary
Clifford+$P_9$
model.
A physical fault-tolerant comparison requires an explicit encoded
stack and a calibrated noise model.

A carrier that accesses
$\ket{2}$
during a fault-tolerant computation must be protected by a qutrit
QECC from the beginning of the computation
\cite{MajumdarHybridFT}.
We therefore compare:

\begin{enumerate}
\item a binary-only circuit protected by the binary
$[[7,1,3]]_2$
Steane code;
\item an all-qutrit circuit protected by a seven-qutrit
Steane-type code.
\end{enumerate}

The comparison below is conditional.
It does not claim a hardware-independent physical advantage.

\subsection{Binary-only reference chains}

A standard clean-ancilla binary decomposition of an
$n$-control MCT gate contains
\begin{equation}\nonumber
2n-3
\end{equation}
ordinary three-qubit Toffoli blocks \cite{Barenco1995}.

Using an exact unitary seven-$T$
Toffoli decomposition \cite{Selinger2013} gives
\begin{equation}\nonumber
\boxed{
N_{T,7}^{\mathrm{bin}}(n)
=
7(2n-3).
}
\label{eq:binary-seven-t}
\end{equation}

A measurement-assisted four-$T$
Toffoli construction gives the alternative benchmark \cite{Jones2013}
\begin{equation}\nonumber
\boxed{
N_{T,4}^{\mathrm{bin}}(n)
=
4(2n-3).
}
\label{eq:binary-four-t}
\end{equation}

The second benchmark uses different ancillary,
measurement,
and feed-forward assumptions.
It is therefore reported separately.

For seven controls:
\begin{equation}\nonumber
N_{T,7}^{\mathrm{bin}}(7)
=
77,
\end{equation}
and
\begin{equation}\nonumber
N_{T,4}^{\mathrm{bin}}(7)
=
44.
\end{equation}

For fifteen controls:
\begin{equation}\nonumber
N_{T,7}^{\mathrm{bin}}(15)
=
189,
\end{equation}
and
\begin{equation}\nonumber
N_{T,4}^{\mathrm{bin}}(15)
=
108.
\end{equation}

These
$T$
counts cannot be compared directly with
$P_9$
counts without specifying encoded resource-state costs.

\subsection{Encoded resource functions}

Let
\begin{equation}\nonumber
\mathcal{F}_{P_9}
\left(
\varepsilon_{P_9},
p_3,
k_3
\right)
\end{equation}
denote the encoded carrier-cycle cost of preparing and injecting one
logical
$P_9$
resource state.

Let
\begin{equation}\nonumber
\mathcal{F}_{T}
\left(
\varepsilon_T,
p_2,
k_2
\right)
\end{equation}
denote the corresponding cost of one logical binary
$T$
resource state.

Here:

\begin{itemize}
\item
$p_2$
is the physical error probability of a binary location;
\item
$p_3$
is the physical error probability of a qutrit location;
\item
$k_2$
is the number of binary concatenation levels;
\item
$k_3$
is the number of ternary concatenation levels.
\end{itemize}

Write
\begin{equation}\nonumber
p_3
=
\delta p_2,
\qquad
\delta\geq1.
\end{equation}

For an allowed failure probability
$\varepsilon_{\MCT}$
for one logical MCT invocation, a conservative union-bound allocation
gives
\begin{equation}\nonumber
\varepsilon_{P_9}
=
\frac{
\varepsilon_{\MCT}
}{
6n+3
},
\end{equation}
and
\begin{equation}\nonumber
\varepsilon_T^{(q)}
=
\frac{
\varepsilon_{\MCT}
}{
q(2n-3)
},
\qquad
q\in\{4,7\}.
\end{equation}

The encoded qutrit-tree cost is
\begin{align}
V_{\mathrm{tree}}(n)
={}&
(6n+3)
\,
\mathcal{F}_{P_9}
\left(
\frac{
\varepsilon_{\MCT}
}{
6n+3
},
p_3,
k_3
\right)
\nonumber
\\
&
+
V_{\mathrm{rest}}^{(3)}(n). 
\label{eq:qutrit-volume}
\end{align}

The binary reference cost is
\begin{align}
V_{\mathrm{bin}}^{(q)}(n)
={}&
q(2n-3)
\,
\mathcal{F}_{T}
\left(
\frac{
\varepsilon_{\MCT}
}{
q(2n-3)
},
p_2,
k_2
\right)
\nonumber
\\
&
+
V_{\mathrm{rest}}^{(2,q)}(n),
\qquad
q\in\{4,7\}.
\label{eq:binary-volume}
\end{align}

The remainder terms include encoded Clifford operations,
syndrome extraction,
routing,
measurement,
reset,
feed-forward,
and ancillary preparation.

The qutrit tree is favorable whenever
\begin{equation}\nonumber
V_{\mathrm{tree}}(n)
<
V_{\mathrm{bin}}^{(q)}(n).
\end{equation}

\subsection{Magic-state-dominated sensitivity conditions}

In a first-order magic-state-dominated approximation,
neglect the remainder terms and approximately match the per-state
output accuracies.

Against the seven-$T$
unitary benchmark,
the qutrit tree is favorable when
\begin{equation}\nonumber
\boxed{
\frac{
\mathcal{F}_{P_9}
}{
\mathcal{F}_{T}
}
<
\frac{
7(2n-3)
}{
6n+3
}.
}
\label{eq:seven-t-crossover}
\end{equation}

Against the four-$T$
measurement-assisted benchmark,
the condition is
\begin{equation}\nonumber
\boxed{
\frac{
\mathcal{F}_{P_9}
}{
\mathcal{F}_{T}
}
<
\frac{
4(2n-3)
}{
6n+3
}.
}
\label{eq:four-t-crossover}
\end{equation}

For seven controls:
\begin{equation}\nonumber
\frac{
7(2n-3)
}{
6n+3
}
=
\frac{77}{45}
\approx
1.711,
\end{equation}
whereas
\begin{equation}\nonumber
\frac{
4(2n-3)
}{
6n+3
}
=
\frac{44}{45}
\approx
0.978.
\end{equation}

For fifteen controls:
\begin{equation}\nonumber
\frac{
7(2n-3)
}{
6n+3
}
=
\frac{189}{93}
\approx
2.032,
\end{equation}
whereas
\begin{equation}\nonumber
\frac{
4(2n-3)
}{
6n+3
}
=
\frac{108}{93}
\approx
1.161.
\end{equation}

Asymptotically:
\begin{equation}\nonumber
\lim_{n\rightarrow\infty}
\frac{
7(2n-3)
}{
6n+3
}
=
\frac{7}{3},
\end{equation}
and
\begin{equation}\nonumber
\lim_{n\rightarrow\infty}
\frac{
4(2n-3)
}{
6n+3
}
=
\frac{4}{3}.
\end{equation}

These inequalities are sensitivity conditions, not physical hardware
predictions.

\subsection{Illustrative sensitivity point}

As an illustrative point, take
\begin{equation}\nonumber
n=15,
\end{equation}
equal binary and ternary concatenation levels, and
\begin{equation}\nonumber
\mathcal{F}_{P_9}
=
1.5
\mathcal{F}_{T}.
\end{equation}

Against the seven-$T$
unitary benchmark:
\begin{equation}\nonumber
V_{\mathrm{tree}}^{\mathrm{magic}}
=
93
\times
1.5
\mathcal{F}_{T}
=
139.5
\mathcal{F}_{T},
\end{equation}
whereas
\begin{equation}\nonumber
V_{\mathrm{bin},7}^{\mathrm{magic}}
=
189
\mathcal{F}_{T}.
\end{equation}

The relative reduction is
\begin{equation}\nonumber
1-
\frac{
139.5
}{
189
}
\approx
0.262.
\end{equation}

Thus, under this illustrative assumption, the qutrit tree reduces
the magic-state-dominated cost by approximately
\begin{equation}\nonumber
26.2\%.
\end{equation}

The same point does not beat the four-$T$
measurement-assisted benchmark.
The comparison therefore depends on the selected binary fault-tolerant
stack.

\subsection{Concatenation-level relation}

The qutrit physical error rate may exceed the qubit physical error
rate.
Following the concatenated-code framework of
\cite{MajumdarHybridFT},
let
$1/c_2$
and
$1/c_3$
denote effective binary and ternary thresholds.

To reach a common logical accuracy,
\begin{equation}\nonumber
\boxed{
k_3
=
\left\lceil
k_2
+
\log_2
\left[
\frac{
\log(c_2p_2)
-
2^{-k_2}
\log(c_2/c_3)
}{
\log(\delta)
+
\log(c_3p_2)
}
\right]
\right\rceil.
}
\label{eq:concatenation-relation}
\end{equation}

In the equal-threshold case,
$c_2=c_3=c$,
this reduces to
\begin{equation}\nonumber
\boxed{
k_3
=
\left\lceil
k_2
+
\log_2
\left[
\frac{
\log(cp_2)
}{
\log(\delta)
+
\log(cp_2)
}
\right]
\right\rceil.
}
\label{eq:equal-threshold-relation}
\end{equation}

These expressions apply only below threshold:
\begin{equation}\nonumber
c_2p_2<1,
\qquad
c_3p_3<1.
\end{equation}

An additional ternary concatenation level may erase the logical
gate-count advantage.
A submission-grade physical-overhead estimate therefore requires
explicit encoded ternary Clifford gates,
$P_9$
state factories,
syndrome-extraction circuits,
and a calibrated qutrit noise model.

\subsection{Transversality caveat}

For the binary Steane code,
logical Clifford operations such as CNOT and Hadamard are
transversal, whereas the logical
$T$
gate is non-transversal.

For the all-qutrit stack,
the encoded implementation must keep the following costs explicit:

\begin{enumerate}
\item ternary SUM between encoded qutrit blocks;
\item
$C_2(\INC)^{\pm1}$;
\item
$C_2(X_{01})$;
\item qutrit syndrome extraction;
\item logical
$P_9$
state preparation and injection.
\end{enumerate}

The logical Clifford+$P_9$ counts in the main text are exact. The physical Steane-code comparison remains conditional on the selected encoded gates and physical noise model.

\end{document}